\documentclass{article}

\usepackage[final]{nips_2018}

\usepackage[utf8]{inputenc} % allow utf-8 input
\usepackage[T1]{fontenc}    % use 8-bit T1 fonts
\usepackage[bottom]{footmisc}
\usepackage{hyperref}       % hyperlinks
\usepackage{url}            % simple URL typesetting
\usepackage{booktabs}       % professional-quality tables
\usepackage{amsfonts}       % blackboard math symbols
\usepackage{nicefrac}       % compact symbols for 1/2, etc.
\usepackage{microtype}      % microtypography
\usepackage{amsmath}
\usepackage{color}
\usepackage{natbib}
\usepackage[ruled,linesnumbered,vlined]{algorithm2e}
\usepackage{multirow}
\usepackage{multicol}
\usepackage{subfig}
\usepackage{graphicx}

\newcommand{\eps}{\epsilon}
\newcommand{\E}{\mathbf{E}}
\renewcommand{\Pr}{\mathbf{Pr}}
\newcommand{\abs}[1]{\left| #1 \right|}

\newcommand{\opt}{{\sf OPT}^{\sf med}}
\newcommand{\optm}{{\sf OPT}^{\sf mea}}
\newcommand{\sol}{{\sf SOL}^{\sf med}}

\newcommand{\loss}{\text{loss}}

\newcommand{\G}{\mathcal{P}}
\newcommand{\zzp}{\mathbb{Z}^+}
\newcommand{\rrp}{\mathbb{R}^+}

\newcommand{\summary}{\textbf{Summary-Outliers}}

\newtheorem{theorem}{Theorem}

\newtheorem{lemma}{Lemma}
\newtheorem{claim}{Claim}

\newtheorem{fact}{Fact}

\newtheorem{corollary}{Corollary}

\newtheorem{definition}{Definition}%}

\newenvironment{proof}{\trivlist\item[]\emph{Proof:}}%
{\unskip\nobreak\hskip 1em plus 1fil\nobreak$\Box$
\parfillskip=0pt%
\endtrivlist}

\newcommand{\kmeansmm}{{\tt $k$-means{-}{-}}}
\newcommand{\kmeanspp}{{\tt $k$-means{+}{+}}}
\newcommand{\gauss}{{\tt gauss}}
\newcommand{\susy}{{\tt  susy}}
\newcommand{\spatial}{{\tt Spatial}}

\newcommand{\kddfull}{{\tt kddFull}}
\newcommand{\kddsample}{{\tt kddSp}}
\newcommand{\lone}{{\tt $\ell_1$-loss}}
\newcommand{\ltwo}{{\tt $\ell_2$-loss}}

\newcommand{\summarysize}{{\tt \small{summarySize}}}

\newcommand{\prerecall}{{\tt preRec}}
\newcommand{\recall}{{\tt recall}}
\newcommand{\precision}{{\tt prec}}
\newcommand{\bg}{{\tt ball-grow}}
\newcommand{\rand}{{\tt rand}}
\newcommand{\kmeanspar}{{\tt $k$-means\textbardbl}}

\newcommand{\qinhides}[1]{}

\DeclareMathOperator{\argmax}{argmax} % no space, limits underneath in displays

\title{A Practical Algorithm for Distributed Clustering and Outlier Detection\ifdefined}
\author{
  Jiecao Chen \\%\thanks{Use footnote for providing further
  Indiana University Bloomington\\
  Bloomington, IN \\
  \texttt{jiecchen@indiana.edu} \\
  \And
  Erfan Sadeqi Azer \\
  Indiana University Bloomington\\
  Bloomington, IN \\
  \texttt{esadeqia@indiana.edu}\\
  \AND
  Qin Zhang\\
  Indiana University Bloomington\\
  Bloomington, IN \\
  \texttt{qzhangcs@indiana.edu} \\
}

\begin{document}
\maketitle

\begin{abstract}
We study the classic $k$-means/median clustering, which are fundamental problems in unsupervised learning, in the setting where data are partitioned across multiple sites, and where we are allowed to discard a small portion of the data by labeling them as outliers.  We propose a simple approach based on constructing small summary for the original dataset. The proposed method is time and communication efficient, has good approximation guarantees, and can identify the global outliers effectively.  
To the best of our knowledge, this is the first practical algorithm with theoretical guarantees for distributed clustering with outliers.  
Our experiments on both real and synthetic data have demonstrated the clear superiority of our algorithm against all the baseline algorithms in almost all metrics.
%\chensays{add synthetic data back?}
\end{abstract}

\section{Introduction}
\label{sec:intro}

The rise of big data has brought the design of distributed learning algorithm to the forefront.  For example, in many practical settings the large quantities of data are collected and stored at different locations, while we want to learn properties of the union of the data.  For many machine learning tasks, in order to speed up the computation we need to partition the data into a number of machines for a joint computation.  In a different dimension, since real-world data often contain background noise or extreme values, it is desirable for us to perform the computation on the ``clean data'' by discarding a small portion of the data from the input.  Sometimes these outliers are interesting by themselves; for example, in the study of statistical data of a population, outliers may represent those people who deserve special attention.  In this paper we study {\em clustering with outliers}, a fundamental problem in unsupervised learning,   in the distributed model where data are partitioned across multiple sites, who need to communicate to arrive at a consensus on the cluster centers 
and labeling of outliers. 

For many clustering applications it is common to model data objects as points in $\mathbb{R}^d$, and the similarity between two objects is represented as the Euclidean distance of the two corresponding points.  In this paper
%we assume for convenience that data items are points in the Euclidean space, though our algorithms work for general metric spaces as well.  
we assume for simplicity that each point can be sent by {\em one unit} of communication. Note that when $d$ is large, we can apply standard dimension reduction tools (for example, the Johnson-Lindenstrauss lemma) before running our algorithms.

We focus on the two well-studied objective functions $(k,t)$-means and $(k,t)$-median, defined in Definition \ref{def:clustering}. It is worthwhile to mention that our algorithms also work for other metrics as long as the distance oracles are given.
\begin{definition}[$(k,t)$-means/median]  
\label{def:clustering}
Let $X$ be a set of points, and $k, t$ be two parameters.  For the $(k,t)$-median problem we aim for computing a set of centers $C \subseteq \mathbb{R}^d$ of size at most $k$ and a set of outliers $O \subseteq X$ of size at most $t$ so that the objective function $\sum_{p\in X\backslash O} d(p, C)$  is minimized.  
For the $(k,t)$-means we simply replace the objective function with $\sum_{p\in X\backslash O} d^2(p, C)$.  
%We call the values of the two objective functions $\ell_1$-loss and $\ell_2$-loss respectively.
\end{definition}

{\bf Computation Model.}  
We study the clustering problems in the {\em coordinator model}, a well-adopted model for distributed learning~\cite{BEL13,CSWZ16,GYZ17,DG0N17}.  In this model we have $s$ sites and a central coordinator; each site can communicate with the coordinator.  The input data points are partitioned among the $s$ sites, who, together with the coordinator, want to jointly compute some function on the global data.   The data partition can be either adversarial or random. The former can model the case where the data points are independently collected at different locations, while the latter is common in the scenario where the system uses a dispatcher to randomly partition the incoming data stream into multiple workers/sites for a parallel processing (and then aggregates the information at a central server/coordinator). 

In this paper we focus on the one-round communication model (also called the {\em simultaneous communication} model), where each site sends a sketch of its local dataset to the coordinator, and then the coordinator merges these sketches and extracts the answer.  This model is arguably the most practical one since multi-round communication will cost a large system overhead.  

Our goals for computing $(k,t)$-means/median in the coordinator model are the following: (1) to minimize the clustering objective functions; (2) to accurately identify the set of global outliers; and (3) to minimize the computation time and the communication cost of the system.  We will elaborate on how to quantify the quality of outlier detection in Section~\ref{sec:exp}.

{\bf Our Contributions.}
A natural way of performing distributed clustering in the simultaneous communication model is to use the two-level clustering framework (see e.g., \cite{GMMMO03,GYZ17}).  In this framework each site performs the first level clustering on its local dataset $X$, getting a subset $X' \subseteq X$ with each point being assigned a weight; we call $X'$ the {\em summary} of $X$. 
The site then sends $X'$ to the coordinator, and the coordinator performs the second level clustering on the union of the $s$ summaries.  We note that the second level clustering is required to output at most $k$ centers and $t$ outliers, while the summary returned by the first level clustering can possibly have more than $(k + t)$ weighted points.  The size of the summary will contribute to the communication cost as well as the running time of the second level clustering.   

The main contribution of this paper is to propose a simple and practical summary construction at sites with the following properties.  
%For simplicity we assume that $k \ge \log n$ where $n$ is the size of the dataset.
\begin{enumerate}
\item It is extremely fast: runs in time $O(\max\{k, \log n\} \cdot n)$, where $n$ is the size of the dataset.

\item The summary has small size: $O(k\log n + t)$ for adversarial data partition and $O(k \log n + t/s)$ for random data partition.

\item When coupled with a second level (centralized) clustering algorithm that $\gamma$-approximates $(k,t)$-means/median, we obtain an $O(\gamma)$-approximation algorithm for distributed $(k,t)$-means/median.\footnote{We say an algorithm $\gamma$-approximates a problem if it outputs a solution that is at most $\gamma$ times the optimal solution.}

\item It can be used to effectively identify the global outliers.
\end{enumerate}
We emphasize that both the first and the second properties are essential to make the distributed clustering algorithm scalable on large datasets. %, and our bounds are the best that one can think of.
%As far as we have concerned, none of the previous algorithms can effectively achieve all of the criteria above. 
Our extensive set of experiments have demonstrated the clear superiority of our algorithm against all the baseline algorithms in almost {\em all} metrics.

To the best of our knowledge, this is the first practical algorithm with theoretical guarantees for distributed clustering with outliers.

{\bf Related Work.}
Clustering is a fundamental problem in computer science and has been studied for more than fifty years.  A comprehensive review of the work on $k$-means/median is beyond the scope of this paper, and we will focus on the literature for centralized/distributed $k$-means/median clustering {\em with} outliers and distributed $k$-means/median clustering.

In the centralized setting, several $O(1)$-approximation or $(O(1), O(1))$-approximation\footnote{We say a solution is an $(a, b)$-approximation if the cost of the solution is $a \cdot C$ while excluding $b \cdot t$ points, where $C$ is the cost of the optimal solution excluding $t$ points.} algorithms have been proposed \cite{CKMN01,C09}. These algorithms make use of linear programming and need time at least $\Omega(n^3)$, which is prohibitive on large datasets.  \citet{FS12} studied $(k,t)$-median via {\em coresets}, but the running times of their algorithm includes a term $O(n(k+t)^{k+t}))$ which is not practical.

\citet{CG13} proposed for $(k,t)$-means an algorithm called \kmeansmm, which is an iterative procedure and can be viewed as a generalization of Llyod's algorithm~\cite{L82}.  Like Llyod's algorithm, the centers that \kmeansmm\ outputs are not the original input points; we thus cannot use it for the summary construction in the first level clustering at sites because some of the points in the summary will be the outliers we report at the end.  
However, we have found that \kmeansmm\ is a good choice for the second level clustering: it outputs exactly $k$ centers and $t$ outliers, and its clustering quality looks decent on datasets that we have tested, though it does not have any worst case theoretical guarantees. 

%\chensays{I remember that Llyod's algorithm has some theoretical guarantee if it is seeded in a certain way?   Here, maybe just say \kmeansmm\ does not have worst case guarantee, do not mention Llyod's}

Recently \citet{GKL+17} proposed a local-search based $(O(1), O(k \log(n))$-approximation algorithm for $(k,t)$-means. The running time of their algorithm is $\tilde{O}(k^2 n^2)$,\footnote{$\tilde{O}(\cdot)$ hides some logarithmic factors.}  which is again not quite scalable. The authors mentioned that one can use the \kmeanspp\ algorithm \cite{AV07} as a seeding step to boost the running time to $\tilde{O}(k^2(k+t)^2 + nt)$. We note that first, this running time is still worse than ours.  And second, since in the first level clustering we only need a summary -- all that we need is a set of weighted points that can be fed into the second level clustering at the coordinator, we can in fact directly use \kmeanspp\ with a budget of $O(k\log n + t)$ centers for constructing a summary.  
%\chensays{Do we need to mention that \kmeanspp\ actually has some theoretical guarantee for $(k, t)$-means}.
We will use this approach as a baseline algorithm in our experimental studies.

In the past few years there has been a growing interest in studying $k$-means/median clustering in the distributed models~\cite{EIM11,BMV+12,BEL13,LBKW14,CEMMP15,CSWZ16}.  In the case of allowing outliers, Guha et al. \cite{GYZ17} gave a first theoretical study for distributed $(k,t)$-means/median. However, their algorithms need $\Theta(n^2)$ running time at sites and are thus again not quite practical on large-scale datasets.  
We note that the \kmeanspar\ algorithm proposed by \citet{BMV+12} can be 
%seen as a parallel version of \kmeanspp\ and thus can also be 
extended (again by increasing the budget of centers from $k$ to $O(k\log n+t)$) and used as a baseline algorithm for comparison. The main issue with  \kmeanspar\ is that it needs $O(\log n)$ rounds of communication which holds back its overall performance.

\ifdefined\submission
\else

\section{Preliminaries}
\label{sec:prelim}

We are going to use the notations listed in Table~\ref{tab:notation}.
\begin{table}[!htb]
\centering
\begin{tabular}{|c|c|c|c|} 
\hline
$X$ & input dataset & $n$ & $n = \abs{X}$, size of the dataset\\
\hline
$k$ & number of centers & $\kappa$ & $\kappa = \max\{k, \log n\}$ \\
\hline
$t$ & number of outliers & $O^*$ & outliers chosen by OPT \\
\hline
$\sigma$ & clustering mapping $\sigma : X \to X$ & $d(y, X)$ & $d(y, X) = \min_{x \in X} d(y, x)$ \\ 
\hline
$\phi_X(\sigma)$ & $\phi_X(\sigma) = \sum_{x \in X} d(x, \sigma(x))$ &
$\phi(X,Y)$ & $\phi(X,Y) = \sum_{y \in Y} d(y, X)$ \\
\hline
$B(S, X, \rho)$ & $=\{x \in X|d(x, S) \le \rho\}$ & 
$r$  & $\#$ of iterations in Algo~\ref{alg:summary} \\
\hline
$X_i$ & remaining points at the $i$-th \\ & iteration  of Algorithm~\ref{alg:summary} &
$W_i$ & $X_i\backslash O^*$ \\                                       
\hline
$C_i$ & clustered points at the $i$-th \\ & iteration of Algorithm~\ref{alg:summary} &
$D_i$ & $C_i\backslash O^*$ \\
\hline
  $\opt_{k,t}(X)$ &  $\min\limits_{\substack{O \subseteq X, \abs{C} \le k \\ \abs{O} \le t}}\sum\limits_{p\in X\backslash O} d(p, C)$ &
$\optm_{k,t}(X)$ & $\min\limits_{\substack{O \subseteq X, \abs{C} \le k\\ \abs{O} \le t}} \sum\limits_{p\in X\backslash O} d^2(p, C)$ \\
\hline
\end{tabular}
\caption{List of Notations}
\label{tab:notation}
\end{table}

We will also make use of the following lemmas.

\begin{lemma}[Chernoff Bound]
\label{lem:Chernoff}
Let $X_1, \ldots, X_n$ be independent Bernoulli random variables such that $\Pr[X_i = 1] = p_i$. Let $X = \sum_{i \in [n]} X_i$, and let $\mu = \E[X]$. It holds that $\Pr[X \ge (1+\delta)\mu] \le e^{-\delta^2\mu/3}$ and $\Pr[X \le (1-\delta)\mu] \le e^{-\delta^2\mu/2}$ for any $\delta \in (0,1)$.
\end{lemma}

   \begin{lemma}[\citet{MP02}]
    \label{lem:bin-ball}
    Consider the classic balls and bins experiment where $b$ balls are thrown into $m$ bins, for some $b, m\in \zzp$. Also, let $w_i$ be a weight associated with the $i$-th bin, for $i\in [m]$. Assuming, the probability of each ball falling into the $i$-th bin is $\frac{w_i}{\sum_{j=1}^m w_j}$ and $b\geq m$, the following holds:
    
    For any $\eps\in \rrp$, there exists a $\gamma\in \rrp$ such that
    $$ \textstyle \Pr[\text{total weight of empty bins} > \eps \sum_{i}w_i]  \leq e^{-\gamma b}.$$
    Note that the dependence of $\gamma$ on $\epsilon$ is independent of $b$ or $m$.
  \end{lemma}

\fi

\section{The Summary Construction}
\label{sec:summary}

In this section we present our summary construction for $(k,t)$-median/means in the centralized model. In Section~\ref{sec:distributed} we will show how to use this summary construction for solving the problems in the distributed model.  
\ifdefined\submission
Table~\ref{tab:notation} is the list of  notations we are going to use.
\begin{table}[!htb]
\centering
\begin{tabular}{|c|c|c|c|} 
\hline
$X$ & input dataset & $n$ & $n = \abs{X}$, size of the dataset\\
\hline
$k$ & number of centers & $\kappa$ & $\kappa = \max\{k, \log n\}$ \\
\hline
$t$ & number of outliers & $O^*$ & outliers chosen by OPT \\
\hline
$\sigma$ & clustering mapping $\sigma : X \to X$ & $d(y, X)$ & $d(y, X) = \min_{x \in X} d(y, x)$ \\ 
\hline
$\phi_X(\sigma)$ & $\phi_X(\sigma) = \sum_{x \in X} d(x, \sigma(x))$ &
$\phi(X,Y)$ & $\phi(X,Y) = \sum_{y \in Y} d(y, X)$ \\
\hline
  $\opt_{k,t}(X)$ &  $\min\limits_{\substack{O \subseteq X, \abs{C} \le k \\ \abs{O} \le t}}\sum\limits_{p\in X\backslash O} d(p, C)$ &
$\optm_{k,t}(X)$ & $\min\limits_{\substack{O \subseteq X, \abs{C} \le k\\ \abs{O} \le t}} \sum\limits_{p\in X\backslash O} d^2(p, C)$ \\
\hline
\end{tabular}
\caption{List of Notations}
\label{tab:notation}
\end{table}

\fi

\subsection{The Algorithm}
\label{sec:algo}

\begin{algorithm}[t]
 \DontPrintSemicolon
  \SetKwInOut{Input}{Input}
  \SetKwInOut{Output}{Output}
  \Input{dataset $X$, number of centers $k$, number of outliers $t$}
  \Output{a weighted dataset $Q$ as a summary of $X$}

  $i \gets 0$,  $X_i \gets X$, $Q \gets \emptyset$\;
  
  fix a $\beta$ such that $0.25 \le \beta < 0.5$\;
  
  $\kappa \gets \max\{\log n, k\}$\;
  
  let $\sigma: X \rightarrow X$ be a mapping to be constructed, and $\alpha$ be a constant to be determined in the analysis. \label{ln:alpha} \;
  %(Lemma \ref{lem:nu-mu}).\;
  
  \While{$|X_i| > 8 t$}{ 
    construct a set $S_i$ of size $\alpha \kappa$ by random sampling (with replacement) from $X_i$ \label{ln:sample} \;
    
    for each point in $X_i$, compute the distance to its nearest point in $S_i$ \label{ln:assign} \;
    
    let $\rho_i$ be the smallest radius s.t.\ $|B(S_i, X_i, \rho_i)| \geq \beta |X_i|$. Let $C_i \gets B(S_i, X_i, \rho_i)$ \label{ln:construction} \; 
    
    for each $x \in C_i$, choose the point $y \in S_i$ that minimizes $d(x, y)$ and assign $\sigma(x) \gets y$ \label{ln:assign-2}  \;
    
    $X_{i+1} \gets X_i \backslash C_i$ \label{ln:prune} \;
    
    $i \gets i + 1$\;
    
  }
  
  $r \gets i$ \;
  % \tcc{$r$ is the number of rounds} \;
  
  for each $x \in X_r$, assign $\sigma(x) \gets x$ \label{ln:outlier}\;

  for each $x \in X_r \cup (\cup_{i=0}^{r-1} S_i)$, assign weight $w_x \gets |\sigma^{-1}(x)|$ and add $(x, w_x)$ into $Q$\;
  
  \Return $Q$\;
  \caption{\summary$(X, k, t)$}
  \label{alg:summary}
\end{algorithm}

Our algorithm is presented in Algorithm~\ref{alg:summary}. It works for both the $k$-means and $k$-median objective functions.   
We note that Algorithm~\ref{alg:summary} is partly inspired by the algorithm for clustering {\em without} outliers proposed in \cite{MP02}. But since  we have to handle outliers now, the design and analysis of our algorithm require new ideas. 

For a set $S$ and a scalar value $\rho$, define $B(S, X, \rho) = \{x \in X\ |\ d(x, S) \le \rho\}$.  Algorithm~\ref{alg:summary} works in rounds indexed by $i$.  Let $X_0 = X$ be the initial set of input points. The idea is to sample a set of points $S_i$ of size $\alpha k$ for a constant $\alpha$ (assuming $k \ge \log n$) from $X_i$, and grow a ball of radius $\rho_i$ centered at each $s \in S_i$. Let $C_i$ be the set of points in the union of these balls.  The radius $\rho_i$ is chosen such that at least a constant fraction of points of $X_i$ are in $C_i$. 

Define $X_{i+1} = X_i \backslash C_i$.  In the $i$-th round, we add the $\alpha k$ points in $S_i$ to the set of centers, and assign points in $C_i$ to their nearest centers in $S_i$. We then recurse on the rest of the points $X_{i+1}$, and stop until the number of points left unclustered becomes at most $8t$. Let $r$ be the final value of $i$. Define the weight of each point $x$ in $\cup_{i=0}^{r-1} S_i$ to be the number of points in $X$ that are assigned to $x$, and the weight of each point in $X_{r}$ to be $1$.  Our summary $Q$ consists of points in $X_{r} \cup (\cup_{i=0}^{r-1} S_i)$ together with their weights.

\subsection{The Analysis}
\label{sec:analysis}

We now try to analyze the performance of Algorithm~\ref{alg:summary}.  The analysis will be conducted for the $(k,t)$-median objective function, while the results also hold for $(k,t)$-means; we will discuss this briefly at the end of this section.
\ifdefined\submission
Due to space constraints, all missing proofs in this section can be found in the supplementary material.
\fi

We start by introducing the following concept.  Note that the summary constructed by Algorithm~\ref{alg:summary} is fully determined by the mapping function $\sigma$ ($\sigma$ is also constructed in Algorithm~\ref{alg:summary}).

\begin{definition}[Information Loss]
\label{def:info-loss}

For a summary $Q$ constructed by Algorithm~\ref{alg:summary}, we define the information loss of $Q$  as 
$$\loss(Q) = \phi_X(\sigma).
$$
That is, the sum of distances of moving each point $x \in X$ to the corresponding center $\sigma(x)$ (we can view each outlier as a center itself).
\end{definition}

We will prove the following theorem, which says that the information loss of the summary $Q$ constructed by Algorithm~\ref{alg:summary} is bounded by the optimal $(k,t)$-median clustering cost on $X$. 

\begin{theorem}
  \label{thm:summary}
Algorithm~\ref{alg:summary} outputs a summary $Q$ such that with probability $(1 - 1/n^2)$ we have that
  $$\textstyle \loss(Q) = O\left( \opt_{k,t}(X) \right).$$
  The running time of Algorithm~\ref{alg:summary} is bounded by $O(\max\{\log n, k\} \cdot n)$, and the size of the outputted summary $Q$ is bounded by $O(k\log n + t)$.
\end{theorem}
\ifdefined\submission
The proof of this theorem relies on building an upper bound on $\phi_X(\sigma)$ and a lower bound on $\opt_{k,t}(X)$. Namely, $\phi_X(\sigma) = O(\sum_i\rho_i|D_i|)$ and $\opt_{k,t}(X) = \Omega(\sum_i\rho_i|D_i|)$, where $D_i = C_i \backslash O^*$, where $C_i$ is constructed in the $i$-th round of Algorithm \ref{alg:summary} and $O^*$ is the set of outliers returned by the optimal algorithm. See the detailed proof in the supplementary material.
\else
\fi

As a consequence of Theorem~\ref{thm:summary}, we obtain by  triangle inequality arguments 
%(similar to \cite{GMMMO03}) 
the following corollary that directly characterizes the quality of the summary  in the task of $(k, t)$-median.
\ifdefined\submission
We include a proof in the supplementary material for completeness.
\fi
\begin{corollary}
  \label{cor:summary-quality}
  If we run a  $\gamma$-approximation algorithm for $(k, t)$-median on $Q$, we can obtain a set of centers $C$ and a set of outliers $O$ such that
  $$\phi(X\backslash O, C) = O(\gamma\cdot \opt_{k,t}(X))$$
  with probability $(1 - 1/n^2)$.
\end{corollary}
\ifdefined\submission
\else
\begin{proof}
  Let $\pi : Q \rightarrow Q$ be the mapping returned by the $\gamma$-approximation algorithm for $(k, t)$-median on $Q$; we thus have $\pi(q)=q$ for all $q \in O$ and $\pi(X\backslash O) = C$. Let $\sigma : X \rightarrow X$ be the mapping returned by Algorithm \ref{alg:summary} (i.e. $\sigma$ fully determines $Q$). We have that
  \begin{align*}
    \phi(X\backslash O, C) 
    &\leq \sum_{x \in X} d(x, \pi(\sigma(x)))\\
                           &\leq \sum_{x \in X} \left( d(x, \sigma(x)) + d(\sigma(x), \pi(\sigma(x))) \right)\\
                           &= \sum_{x\in X} d(x, \sigma(x)) +  \sum_{x\in X} d(\sigma(x), \pi(\sigma(x))) \\
                           &= \loss(Q) +  \sum_{q\in Q}w_q \cdot d(q, \pi(q)) \\
    &= \loss(Q) +  \sol_{k,t}(Q),
  \end{align*}
  where $\sol_{k, t}(Q)=\sum_{q\in Q}w_q \cdot d(q, \pi(q))$ denotes the cost of the $\gamma$-approximation on $Q$. The corollary follows from  Theorem \ref{thm:summary} and Lemma \ref{lem:two-level} (set $s = 1$).
\end{proof}
\fi

\ifdefined\submission
\else

In the rest of this section we prove Theorem~\ref{thm:summary}.  We will start by bounding the information loss.

\begin{definition}[$O^*$, $W_i$ and $D_i$]
Define $O^* \subseteq X$ to be the set of outliers chosen by running the optimal $(k, t)$-median algorithm on $X$; we thus have $\abs{O^*} = t$.  For $i = 0, 1, \ldots, r-1$, define $W_i = X_i \backslash O^*$ and $D_i = C_i \backslash O^*$, where $X_i$ and $C_i$ are  defined in Algorithm~\ref{alg:summary}. 
\end{definition}

We need the following utility lemma. It says that at each iteration in the while loop in Algorithm \ref{alg:summary}, we always make sure that at least half of the remaining points are not in $O^*$.

\begin{lemma}
  \label{lem:X-W}
  For any $0 \leq i < r$, where $r$ is the total number of rounds in Algorithm \ref{alg:summary}, we have $2|W_i| \geq |X_i|$.
\end{lemma}
\ifdefined\submission
\else
\begin{proof}
According to the condition of the while loop in Algorithm \ref{alg:summary} we have $|X_i| > 8 t$ for any $0 \leq i < r$. Since $|O^*| = t$, we have 
$$2|W_i| = 2|X_i\backslash O^*| \geq |X_i| + (|X_i| - 2|O^*|) \geq |X_i|.$$
\end{proof}
\fi

The rest of the proof for Theorem~\ref{thm:summary} proceeds as follows.  We first show in Lemma~\ref{lem:ub} that $\loss(Q) = \phi_{X}(\sigma)$ can be upper bounded by $O(\sum_{0 \le i < r} \rho_i \abs{D_i})$ (Lemma~\ref{lem:ub}). We then show in Lemma~\ref{lem:lb} that $\opt_{k,t}(X)$ can be lower bounded by $\Omega(\sum_{0 \le i < r} \rho_i \abs{D_i})$ with high probability (Lemma~\ref{lem:lb}).  Theorem~\ref{thm:summary} then follows.
\ifdefined\submission
As a reminder, all the missing proofs can be found in the supplementary material.
\fi

\begin{lemma}[upper bound]
  \label{lem:ub}
  It holds that
  $$ \phi_{X}(\sigma) \leq 2\sum_{0 \leq i < r} \rho_i |D_i|.$$
Here $\rho_i$ is the radius we chosen in the $i$-th round of Algorithm \ref{alg:summary}.
\end{lemma}

\ifdefined\submission
\else
\begin{proof}
  First, note that by Line~\ref{ln:construction} and the condition of the while loop in Algorithm~\ref{alg:summary} we have
  \begin{equation} 
  \label{eq:a-1}
  |C_i| \geq \beta |X_i| \geq 8\beta t \overset{\beta \geq 0.25}{\geq} 2 t.
  \end{equation} 
  We thus have by the definition of $D_i$ that
  \begin{eqnarray}
    |D_i| &=& |C_i \backslash O^*| \geq |C_i| - |O^*| \nonumber \\
    &\overset{|O^*| = t}{=}& |C_i| - t \nonumber  \\
    &\overset{\text{by } (\ref{eq:a-1})}{\geq}& |C_i| / 2.  \label{eq:a-2}
  \end{eqnarray}

  Observe that $X\backslash X_r = \cup_{0 \leq i < r} C_i$ and $C_i \cap C_j = \emptyset$ for any $i\neq j$, we can bound $\phi_{X\backslash X_r}(\sigma)$ by the following.
  \begin{eqnarray*}
    \phi_{X\backslash X_r}(\sigma) &=& \sum_{0 \leq i < r} \phi_{C_i}(\sigma) \\
    &\leq& \sum_{0 \leq i < r} \rho_i |C_i|\\
    &\overset{\text{by } (\ref{eq:a-2})}{\leq}& \sum_{0 \leq i < r} 2 \rho_i |D_i|.
  \end{eqnarray*}
  The lemma follows since by our construction at Line~\ref{ln:outlier} we have $\phi_{X_r}(\sigma) = 0$.
\end{proof}
\fi

We now turn to the lower bound of $\opt_{k,t}(X)$.  

\begin{lemma}[lower bound]
  \label{lem:lb}
  It holds that
  $$\textstyle \opt_{k,t}(X) = \Omega\left(\sum_{0 \leq i < r} \rho_i |D_i|\right).$$
\end{lemma}

Before proving the lemma, we would like to introduce a few more notations.

\begin{definition}[$\rho_i^{\tt opt}$ and $h$]
  \label{def:rho}
  Let $h = \frac{1 + 2\beta}{2}$; we thus have $1 > h > 2\beta > 0$ (recall in Algorithm \ref{alg:summary} that $\beta < 0.5$ is a fixed constant). For any $0 \leq i < r$, let $\rho_i^{\tt opt} > 0$ be the \emph{minimum} radius such that there exists a set $Y \subseteq X\backslash O^*$ of size $k$ with
  \begin{equation}
    \label{eq:mu}
    |B(Y, W_i, \rho_i^{\tt opt})| \geq h |W_i|.
  \end{equation}
\end{definition}

The purpose of introducing $\rho_i^{\tt opt}$ is to use it as a bridge to connect $\opt_{k,t}(X)$ and $\rho_i$.  We first have the following.

\begin{lemma}
  \label{lem:lb-2}
  $\opt_{k,t}(X) = \Omega\left( \sum_{0 \leq i < r} \rho_i^{\tt opt} |D_i|\right)$.
\end{lemma}

\ifdefined\submission
\else
Fix an arbitrary set $Y \subseteq X \backslash O^*$ of size $k$ as centers.  To prove Lemma~\ref{lem:lb-2} we will use a charging argument 
to connect $\opt_{k,t}(X)$ and $\sum_{0 \leq i < r} \rho_i^{\tt opt} |D_i|$. To this end we introduce the following definitions and facts.

\begin{definition}[$E_i$, $E^m_i$ and $\G^m_\ell$]
For each $0 \leq i < r$, define $E_i = \{x \in W_i ~|~ d(x, Y) \geq \rho_i^{\tt opt}\}$. For any $m \in \mathbb{Z}^+$, define $E^m_i = E_i \backslash(\cup_{j > 0} E_{i + jm})$.  Let $\G_\ell^m = \{ 0 \leq i < r~|~ i \equiv \ell\pmod m\}$.
\end{definition}

Clearly, if $i \neq j$ and $j \equiv i \pmod m$, then $E^m_i$ and $E^m_j$ are disjoint. This leads to the following fact.

 \begin{fact}
    \label{fact:G}
    For any $i = 0, 1, \ldots, r-1$, we have
    \begin{align*}
      \phi(Y, \cup_{i\in \G_\ell^m} E_i^m) &= \sum_{i\in \G_\ell^m}\phi(Y, E_i^m)\\
                                     &\geq \sum_{i\in \G_\ell^m} \rho_i^{\tt opt}|E_i^m|.
    \end{align*}
  \end{fact}

By the definitions of $\rho_i^{\tt opt}$ and $E_i$ we directly have: 
  \begin{fact}
  For any $i = 0, 1, \ldots, r-1$,
    \label{fact:EW}
    $|E_i| \geq (1 - h) |W_i|$.
  \end{fact}
  
Let $z = \lceil \log_{1-\beta} \frac{1-h}{6} \rceil$ (a constant), we have

\begin{fact}
    \label{fact:half}
 For any $i = 0, 1, \ldots, r-1$, $\abs{E_i^z} \ge \abs{E_i}/2$.
\end{fact}

\begin{proof}
We first show that $\abs{E_i}, \abs{E_{i+z}}, \ldots$ is a geometrically decreasing sequence.
\begin{eqnarray*}
    |E_{i + z}|   &\leq& |X_{i+z}| \\
    &\leq& (1 - \beta)^z|X_i|  \\
    &\overset{\text{Lemma \ref{lem:X-W}}}{\leq}& 2(1 - \beta)^z|W_i| \\
    &\overset{\text{Fact \ref{fact:EW}}}{\leq}& \frac{2(1 - \beta)^z}{1 - h}|E_i|\\
    &\overset{\text{Def. of} ~z}{\leq}& \frac{|E_i|}{3}.
    \end{eqnarray*}
As a result, we have that $E_i^z$ holds a least a constant fraction of points in $E_i$.
\begin{eqnarray*}
	|E_i^z| &=& |E_i \backslash \cup_{j > 0} E_{i + jz}| \\
	&\geq& |E_i| - \sum_{j>0} \frac{|E_i|}{3^j} \\ 
	&\geq& \frac{|E_i|}{2}.
\end{eqnarray*}
\end{proof}

\begin{fact}
    \label{fact:all}
 For any $i = 0, 1, \ldots, r-1$, $\abs{E_i^z} \ge  (1-h)\abs{D_i}/2$.
\end{fact}

\begin{proof}
$$\abs{E_i^z}
    \overset{\textbf{Fact}~\ref{fact:half}}{\geq}  |E_i|/2
    \overset{\textbf{Fact}~\ref{fact:EW}}{\geq} (1-h)|W_i|/2
    \overset{D_i \subseteq W_i}{\geq} (1-h)|D_i|/2.  
$$
\end{proof}

\begin{proof}(of Lemma~\ref{lem:lb-2})
Let 
$$ \textstyle \ell = \argmax_{0 \leq j < z} \left( \sum_{i\in \G_j^z} |E_i^z|\right).$$ Then $\phi(Y, X \backslash O^*)$ is at least
  \begin{eqnarray*}
  \phi(Y, \cup_{i \in \G_{\ell}^z} E^z_i)  
                                             &\overset{\textbf{Fact}~\ref{fact:G}}{\geq}& \sum_{i \in \G_{\ell}^z} \rho_i^{\tt opt} |E_i^z| \\
                                             &\overset{\text{def. of\ } \ell}{\geq}& \frac{1}{z} \sum_{0 \leq i < r} \rho_i^{\tt opt} |E^z_i|\\
                                             &\overset{\textbf{Fact}~\ref{fact:all}}{\geq}& \Omega(1) \cdot \sum_{0 \leq i < r} \rho_i^{\tt opt} |D_i|.  
  \end{eqnarray*}
The lemma then follows from the fact that $Y$ is chosen arbitrarily.
\end{proof}
\fi

Note that Lemma \ref{lem:lb-2} is slightly different from Lemma \ref{lem:lb} which is what we need, but we can link them by proving the following lemma.

\begin{lemma}
  \label{lem:nu-mu}
  With probability $1 - 1/n^2$, we have $\rho_i^{\tt opt} \geq {\rho_i}/{2}$ for all $0 \leq i < r$.
\end{lemma}

\ifdefined\submission
\else
\begin{proof}
Fix an $i$, and let $Y \subseteq X\backslash O^*$ be a set of size $k$ such that $|B(Y, W_i, \rho_i^{\tt opt})| \geq h |W_i|$.  Let $G = B(Y, W_i, \rho_i^{\tt opt})$. We assign each point in $G$ to its closest point in $Y$, breaking ties arbitrarily.  Let $P_x$ be the set of all points in $G$ that are assigned to $x$; thus $\{P_x~|~x\in Y\}$ forms a partition of $G$.  

Recall that $S_i$ in Algorithm \ref{alg:summary} is constructed by a random sampling. Define
$$ G' = \{y \in G\ |\ \exists x \in Y\ s.t.\ (y \in P_x) \wedge (S_i\cap P_x \neq \emptyset)\}.
$$
We have the following claim.

\begin{claim}
    \label{claim:g-g}
    For any positive constant $\eps$, there exists a sufficiently large constant $\alpha$ (Line~\ref{ln:alpha} in Algorithm~\ref{alg:summary}) such that
    \begin{equation}
      \label{eq:g-g}
      |G'| \geq (1 - \eps) |G|
    \end{equation}
    with probability $1 - 1/n^2$.  
  \end{claim}

Note that once we have $(\ref{eq:g-g})$, we have that for a sufficiently small constant $\eps$,
  \begin{eqnarray*}
    |G'| &\geq& (1 - \eps) |G| \\
    	&\overset{\text{Def.}\ \ref{def:rho}}{\geq}& (1 - \eps) h |W_i|\\
         & \overset{h > 2\beta}{\geq} & 2\beta |W_i| \\
         & \overset{\text{Lemma } \ref{lem:X-W}}{\geq} & \beta |X_i|.
  \end{eqnarray*}
  Since $G' \subseteq B(S_i, W_i, 2\rho_i^{\tt opt}) \subseteq B(S_i, X_i, 2\rho_i^{\tt opt})$, we have $|B(S_i, X_i, 2\rho_i^{\tt opt})| \geq \beta |X_i|$. By the definition of $\rho_i$, we have that $\rho_i \leq 2 \rho_i^{\tt opt}$.  The success probability $1 - 1/n$ in Lemma~\ref{lem:nu-mu} is obtained by applying a union bound over all $O(\log n)$ iterations.
  
Finally we prove Claim~\ref{claim:g-g}.  By the definition of $G$ and Lemma~\ref{lem:X-W} we have 
\begin{equation}
\label{eq:b-1}
\abs{G} \ge h \abs{W_i} \ge h/2 \cdot \abs{X_i}.
\end{equation}
Denote $S_i = \{s_1, \ldots, s_{\alpha \kappa}\}$.
Since $S_i$ is a random sample of $X_i$ (of size $\alpha \kappa$), by  (\ref{eq:b-1}) we have that for each point $j \in [\alpha \kappa]$, $\Pr[s_j \in G] \ge h/2$.   For each $j \in [\alpha \kappa]$, define a random variable $Y_j$ such that $Y_j = 1$ if $s_j \in G$, and $Y_j = 0$ otherwise.  Let $Y = \sum_{i \in [\alpha \kappa]} Y_j$; we thus have $\E[Y] \ge h/2 \cdot \alpha \kappa$. By applying Lemma~\ref{lem:Chernoff} (Chernoff bound) on $Y_j$'s, we have that for any positive constant $\gamma$ and $h$, there exists a sufficiently large constant $\alpha$ (say, $\alpha = 10h/\gamma^2$) such that 
\begin{eqnarray*}
\Pr[Y \ge \gamma \kappa] & \ge & 1 - e^{- \left(\frac{2\gamma}{h}\right)^2 \cdot \frac{h}{2} \alpha \kappa / 2}\\
& \ge & 1 - 1/n^3,
\end{eqnarray*}
In other words, with probability at least $1 - 1/n^3$,  $\abs{S_i \cap G} \ge \gamma \kappa$.  The claim follows by applying Lemma~\ref{lem:bin-ball} on each point in $S_i \cap G$ as a ball, and each set $P_x$ as a bin with weight $\abs{P_x}$.  
\end{proof}
\fi

Lemma~\ref{lem:lb} follows directly from Lemma~\ref{lem:lb-2} and Lemma~\ref{lem:nu-mu}.

\fi

{\bf The running time.}  We now analyze the running time of Algorithm~\ref{alg:summary}.

At the $i$-th iteration, the sampling step at Line~\ref{ln:sample} can be done in $O(\abs{X_i})$ time.  The nearest-center assignments at Line~\ref{ln:assign} and \ref{ln:assign-2} can be done in $\abs{S_i} \cdot \abs{X_i} = O(\kappa \abs{X_i})$ time.  Line~\ref{ln:construction} can be done by first sorting the distances in the increasing order and then scanning the shorted list until we get enough points. In this way the running time is bounded by $\abs{X_i} \log \abs{X_i} = O(\kappa \abs{X_i})$. Thus the total running time can be bounded by 
$$
\sum_{i = 0, 1, \ldots, r - 1} O(\kappa \abs{X_i}) = O(\kappa n) = O(\max\{\log n, k\} \cdot n),
$$
where the first equation holds since the size of $X_i$ decreases geometrically, and the second equation is due to the definition of $\kappa$.

Finally, we comment that we can get a similar result for $(k, t)$-means by appropriately adjusting various constant parameters in the proof.
\ifdefined\submission
Please refer to the supplementary material for a more detailed discussion.
\else
\begin{corollary}
  \label{cor:summary}
Let $X_r$ and $\sigma : X \to X$ be computed by Algorithm~\ref{alg:summary}.
We have with probability $(1 - 1/n^2)$ that
$$ \textstyle \sum_{x\in X\backslash X_r}d^2(x,\sigma(x)) = O\left(\optm_{k,t}(X) \right).$$

\end{corollary}

We note that in the proof for the median objective function we make use of the triangle inequality in various places, while for the means objective function where the distances are squared, the triangle inequality does not hold. However we can instead use the inequality $2(x^2 + y^2) \ge (x + y)^2$, which will only make the constant parameters in the proofs slightly worse. 
\fi

\subsection{An Augmentation}
\label{sec:augment}

In the case when $t \gg k$, which is typically the case in practice since the number of centers $k$ does not scale with the size of the dataset while the number of outliers $t$ does,
we add an augmentation procedure to Algorithm~\ref{alg:summary} to achieve a better practical performance.  The procedure is presented in Algorithm~\ref{alg:augment}.

\begin{algorithm}[t]
 \DontPrintSemicolon
  \SetKwInOut{Input}{Input}
  \SetKwInOut{Output}{Output}
  \Input{dataset $X$, number of centers $k$, number of outliers $t$}
  \Output{a weighted dataset $Q$ as a summary of $X$}

  run {\em Summary-Outliers$(X, k, t)$} (Algorithm~\ref{alg:summary}) and obtain $X_r$ and $S = \cup_{i=0}^{r-1} S_i$\;
  
  construct a set $S'$ of size $\abs{X_r} - \abs{S}$ by randomly sampling (with replacement) from $X \backslash (X_r \cup S)$\;
  
  for each $x \in X \backslash X_r$, set $\pi(x) \gets \arg\min_{y \in S \cup S'} d(x, y)$\;
  
  for each $x \in X_r \cup (\cup_{i=0}^{r-1} S_i)$, assign weight $w_x \gets |\pi^{-1}(x)|$ and add $(x, w_x)$ into $Q$\;
  
  \Return $Q$\;
 
  \caption{\bf Augmented-Summary-Outliers$(X, k, t)$}
  \label{alg:augment}
\end{algorithm}

The augmentation is as follows, after computing the set of outliers $X_r$ and the set of centers $S = \cup_{i=0}^{r-1} S_i$ in Algorithm~\ref{alg:summary}, we sample randomly from $X \backslash (X_r \cup S)$ an additional set of center points $S'$ of size $\abs{X_r} - \abs{S}$.  That is, we try to make the number of centers and the number of outliers in the summary to be balanced.  We then reassign each point in the set $X \backslash X_r$ to its nearest center in $S \cup S'$.  Denote the new mapping by $\pi$. Finally, we include points in $X_r$ and $S$, together with their weights, into the summary $Q$.

It is clear that the augmentation procedure preserves the size of the summary asymptotically. And by including more centers we have 
$\loss(Q) \le \phi_{X}(\pi) \le \phi_{X}(\sigma)$,  
where $\sigma$ is the mapping returned by Algorithm \ref{alg:summary}. The running time will increase to $O(t n)$ due to the reassignment step, but our algorithm is still much faster than all the baseline algorithms, as we shall see in Section~\ref{sec:exp}.

\section{Distributed Clustering with Outliers}
\label{sec:distributed}

In this section we discuss distributed $(k,t)$-median/means using the summary constructed in Algorithm~\ref{alg:summary}. 
\ifdefined\submission
Our main result is the following theorem, which is based on the work by \cite{GMMMO03,GYZ17}. The proof for this theorem can be found in the supplementary material.
\else
We will first discuss the case where the data is randomly partitioned among the $s$ sites, which is the case in all of our experiments. The algorithm is presented in Algorithm~\ref{alg:distributed}. We will discuss the adversarial partition case at the end.   We again only show the results for $(k,t)$-median since the same results will hold for $(k,t)$-means by slightly adjusting the constant parameters. 
\fi

\begin{algorithm}[t]
 \DontPrintSemicolon
  \SetKwInOut{Input}{Input}
  \SetKwInOut{Output}{Output}
  \Input{For each $i \in [s]$, Site $i$ gets input dataset $A_i$ where $(A_1, \ldots, A_s)$ is a random partition of $X$}
  \Output{a $(k,t)$-median clustering for $X = \cup_{i \in [s]} A_i$}

  for each $i \in [s]$, Site $i$ constructs a summary $Q_i$ by running {\em Summary-Outliers$(A_i, k, 2t/s)$} (Algorithm~\ref{alg:summary}) and sends $Q_i$ to the coordinator \label{ln:local} \;
  
  the coordinator then performs a second level clustering on $Q = Q_1 \cup Q_2 \cup \ldots \cup Q_s$ using an {off-the-shelf} $(k,t)$-median algorithm, and returns the resulting clustering.  \label{ln:second}
 
  \caption{\bf Distributed-Median$(A_1, \ldots, A_s, k, t)$}
  \label{alg:distributed}
\end{algorithm}  

\ifdefined\submission
\else
We will make use of the following known results.  The first lemma says that the sum of costs of local optimal solutions that use the same number of outliers as the global optimal solution does is upper bounded by the cost of the global optimal solution.

\begin{lemma}[\cite{GYZ17}]
\label{lem:local-global}

For each $i \in [s]$, let $t_i = \abs{A_i \cap O^*}$ where $O^*$ is the set of  outliers produced by the optimal $(k,t)$-median algorithm on $X = A_1 \cup A_2 \cup \ldots \cup A_s$.  We have
$$ \textstyle \sum_{i \in [s]} \opt_{k, t_i}(A_i) \le O\left( \opt_{k, t}(X) \right).$$
\end{lemma}

The second lemma is a folklore for two-level clustering.

\begin{lemma}[\cite{GMMMO03,GYZ17}]
\label{lem:two-level}

Let $Q = Q_1 \cup Q_2 \cup \ldots \cup Q_s$ be the union of the summaries of the $s$ local datasets, and let $\sol_{k,t}(\cdot)$ be the cost function of a $\gamma$-approximation algorithm for $(k,t)$-median. We have
$$ \textstyle \sol_{k,t}(Q) \le O(\gamma) \cdot \left(\sum_{i \in [s]} \loss(Q_i) + \opt_{k,t}(X)\right).
$$
\end{lemma}

Now by Lemma~\ref{lem:local-global}, Lemma~\ref{lem:two-level} and Theorem~\ref{thm:summary}, we have that with probability $1 - 1/n$,
$\sol_{k,t}(Q) \le O(\gamma) \cdot \opt_{k,t}(X)$.  And by Chernoff bounds and a union bound we have $t_i \le 2t/s$ for all $i$ with probability $1 - 1/n^2$.\footnote{For the convenience of the analysis we have assumed $t/s \ge \Omega(\log n)$, which is justifiable in practice since $t$ typically scales with the size of the dataset while $s$ is usually a fixed number.}  
\fi

\begin{theorem}
\label{thm:main}

Suppose Algorithm~\ref{alg:distributed} uses a $\gamma$-approximation algorithm for $(k,t)$-median in the second level clustering (Line~\ref{ln:second}).  We have with probability $(1 - 1/n)$ that:
\begin{itemize}
\item it outputs a set of centers $C \subseteq \mathbb{R}^d$ and a set of outliers $O \subseteq X$ such that  $\phi(X \backslash O, C) \le O(\gamma) \cdot \opt_{k,t}(X)$;

\item it uses one round of communication whose cost is bounded by $O(sk\log n + t)$;

\item the running time at the $i$-th site is bounded by $O(\max\{\log n, k\} \cdot \abs{A_i})$, and the running time at the coordinator is that of the second level clustering.
\end{itemize}
\end{theorem}

In the case that the dataset is adversarially partitioned, the total communication increases to $O(s(k\log n + t))$.  This is because all of the $t$ outliers may go to the same site and thus $2t/s$ in Line \ref{ln:local} needs to be replaced by $t$.

Finally, we comment that the result above also holds for the summary
constructed in Algorithm~\ref{alg:augment}, except, as discussed in Section~\ref{sec:summary}, that the local running time at the $i$-th site will increase to $O(t \abs{A_i})$.

\section{Experiments}
\label{sec:exp}
%In this section we validate our proposed algorithm on real and synthetic datasets. 

%We will compare our algorithm with other baseline algorithms in the task of $(k,t)$-means/median clustering and outlier detection.  

%Recall that our algorithm only generates a summary (a weighted subset of points) of the original dataset. As mentioned in the introduction, we will apply \kmeansmm\ \cite{CG13} on our summary (and those produced by competitor algorithms as well) as the second level clustering and output $k$ centers and $t$ outliers as the final solution.

%Given the parameters $k$ and $t$, \kmeansmm\ computes $k$ centers and $t$ points from the original dataset as the outliers. We note that in a very recent paper \cite{GKL+17}, an local search based algorithm is proposed for clustering with outliers. However, \cite{GKL+17} may output more than $t$ outliers and therefore it does not serve our purpose here for re-clustering.

\subsection{Experimental Setup}

\subsubsection{Datasets and Algorithms}

\ifdefined\submission
Due to space constraints, we only present the experimental results for two data sets (\kddfull\ and \kddsample). One can find results for a number of other datasets in our supplementary materials and the full paper.
\else
We make use of the following datasets.
\fi
%\erfansays{after colon capitalize all or none}

\begin{itemize}
  \ifdefined\submission
  \else
\item \gauss-$\sigma$. This is a synthetic dataset, generated as follows: we first sample $100$ centers from $[0, 1]^5$, i.e., each dimension is sampled uniformly at random from $[0, 1]$. For each center $c$, we generate $10000$ points by adding each dimension of $c$ a random value sampled from the normal distribution $\mathcal{N}(0, \sigma)$. This way, we obtain $100\cdot 10000 = 1$M points in total. We next construct the outliers as follows: we sample $5000$ points from the $1$M points, and for each sampled point, we add a random shift sampled from $[-2, 2]^5$. %\erfansays{multivariate normal distribution?}
\fi

\item \kddfull. This dataset is from 1999 kddcup competition and contains instances describing connections of sequences of tcp packets. There are about $4.9$M data points.
  \ifdefined\submission
  \else
  \footnote{More information can be found in \url{http://kdd.ics.uci.edu/databases/kddcup99/kddcup99.html}}
  \fi
  We only consider the $34$ numerical features of this dataset. We also normalize each feature so that it has zero mean and unit standard deviation.  There are $23$ classes in this dataset, $98.3\%$ points of the dataset belong to $3$ classes ({\tt normal} $19.6\%$, {\tt neptune} $21.6\%$, and {\tt smurf} $56.8\%$). We consider small clusters as outliers and there are $45747$ outliers.
%\qinsays{Give the url of this dataset as a footnote; same as that for \susy}
%\erfansays{38 or 34}
%\chensays{Erfan, in the kmeans-- paper, it's 38, why it is 34 here?}
  
\item \kddsample. This data set contains about $10\%$ points of \kddfull\ (released by the original provider). This dataset is also normalized and there are 8752 outliers. 

  \ifdefined \submission
  \else
\item \susy-$\Delta$. This data set has been produced using Monte Carlo simulations by \citet{BSW14}. Each instance has $18$ numerical features and there are $5$M instances in total.\footnote{More information about this dataset can be found in \url{https://archive.ics.uci.edu/ml/datasets/SUSY}}. We normalize each feature as we did in \kddfull. We manually add outliers as follows: first we randomly sample $5000$ data points; for each data point, we shift each of its dimension by a random value chosen from $[-\Delta, \Delta]$.
  
  \item \spatial-$\Delta$. This dataset is about 3D road network with elevation information from North Jutland, Denmark. It is designed for clustering and regression tasks. There are about 0.4M data points with 4 features. We normalize each feature so that it has zero mean and unit standard deviation. We add outliers as we did for \susy-$\Delta$.
\footnote{More information can be found in \url{https://archive.ics.uci.edu/ml/datasets/3D+Road+Network+(North+Jutland,+Denmark)}.}

% \item \power-$\Delta$. This dataset is about electric measurements on the power consumption of a household over a period of almost 4 years. There are about 2M data points and used 7 numerical features in our experiment. We normalize each feature so that it has zero mean and unit standard deviation. We add outliers as we did for \susy-$\Delta$.
% \footnote{More information can be found in \url{https://archive.ics.uci.edu/ml/datasets/individual+household+electric+power+consumption}.}
\fi
\end{itemize}

\ifdefined\submission
We comment that finding appropriate $k$ and $t$ values for the task of clustering with outliers is a separate problem, and is not part of the topic of this paper. In all our experiments, $k$ and $t$ are naturally suggested by the datasets we use.
\else
Finding appropriate $k$ and $t$ values for the task of clustering with outliers is a separate problem, and is not part of the topic of this paper. In all our experiments, $k$ and $t$ are naturally suggested by the datasets we use unless they are unknown.
\fi

%Note that when calculate above metrics, we ignore all the weights. \lone\ and \ltwo\ are used to measure the quality of the summary in the task of clustering with outliers and the remaining metrics are mainly used to evaluate the performance in the task of outliers detection. 
%\qinsays{I think this paragraph will only confuse people. Remove?}

%\subsubsection{Algorithms}
%\label{sec:alg}

We compare the performance of following algorithms, each of which is implemented using the MPI framework and run in the coordinator model. The data are randomly partitioned among the sites.
\begin{itemize}
\item \bg. Algorithm~\ref{alg:distributed} proposed in this paper, with the augmented version Algorithm~\ref{alg:summary} for the summary construction.  As mentioned we use \kmeansmm\ as the second level clustering at Line~\ref{ln:second}.    
We fix $\alpha = 2$ and $\beta=4.5$ in the subroutine Algorithm~\ref{alg:summary}.

%\qinsays{Is ball-grow a good name? change value to radius in Alg 1}

\item \rand. Each site constructs a summary by randomly sampling points from its local dataset. Each sampled point $p$ is assigned a weight equal to the number of points in the local dataset that are closer to $p$ than other points in the summary. The coordinator then collects all weighted samples from all sites and feeds to \kmeansmm\ for a second level clustering. 

\item \kmeanspp. Each site constructs a summary of the local dataset using the \kmeanspp\ algorithm \cite{AV07}, and sends it to the coordinator. The coordinator feeds the unions all summaries to \kmeansmm\ for a second level clustering. 

\item \kmeanspar. An MPI implementation of the \kmeanspar\ algorithm proposed by \citet{BMV+12} for distributed $k$-means clustering. To adapt their algorithm to solve the outlier version, we increase the parameter $k$ in the algorithm to $O(k+t)$, and then feed the outputted centers to \kmeansmm\ for a second level clustering. 

%This algorithm proceeds in multiple rounds. In each round, each point $p$ will be included in the summary with probability $\frac{\ell\cdot D_p^2}{\sum_{q} D_q^2}$. Here $D_q$ is the minimum distance from $q$ to points sampled in all previous rounds, $\ell$ is a parameter taken by this algorithm and it controls the size of the final summary.
\end{itemize}
%Recall that \kmeansmm\ outputs a set of centers $C$ and a set of weighted points $O \subset U$, where $|C| = k$ and the total weight of points in $O$ is equal to $t$ (note that $|O|$ may be smaller than $t$ since points may have weights).

\subsubsection{Measurements}
Let $C$ and $O$ be the sets of centers and outliers respectively returned by a tested algorithm.  To evaluate the quality of the clustering results we use two metrics:
(a) \lone\ (for $(k,t)$-median): $\sum_{p \in X\backslash O} d(p, C)$;
(b) \ltwo\  (for $(k,t)$-means): $\sum_{p \in X\backslash O} d^2(p, C)$.

To measure the performance of outlier detection we use three metrics. Let $S$ be the set of points fed into the second level clustering \kmeansmm\ in each algorithm, and let $O^*$ be the set of actual outliers (i.e., the ground truth), we use the following metrics: (a) \prerecall: the proportion of actual outliers that are included in the returned summary, defined as $\frac{|S\cap O^*|}{|O^*|}$;
(b) \recall: the proportion of actual outliers that are returned by \kmeansmm, defined as $\frac{|O\cap O^*|}{|O^*|}$;
(c) \precision: the proportion of points in $O$ that are actually outliers, defined as $\frac{|O\cap O^*|}{|O|}$.

%\ifdefined\submission
%\else
\subsubsection{Computation Environments}
All algorithms are implemented in C++ with Boost.MPI support. We use Armadillo \cite{C10} as the numerical linear library and -O3 flag is enabled when compile the code. All experiments are conducted in a PowerEdge R730 server equipped with 2 x Intel Xeon E5-2667 v3 3.2GHz. This server has 8-core/16-thread per CPU, 192GB Memeory and 1.6TB SSD.
%\fi

\subsection{Experimental Results}
We now present our experimental results. All results take the average of $10$ runs.
\ifdefined\submission
In our supplementary material, results for more datasets can be found, but all the conclusions remain the same.
\fi
\subsubsection{Quality}
We first compare the qualities of the summaries returned by \bg, \rand\ and \kmeanspar. Note that the size of the summary returned by \bg\ is determined by the parameters $k$ and $t$, and we can not control the exact size. In \kmeanspar, the summary size is determined by the sample ratio, and again we can not control the exact size. On the other hand, the summary sizes of \rand\ and \kmeanspp\ can be fully controlled. To be fair, we manually tune those parameters so that the sizes of summaries returned by different algorithms are roughly the same (the difference is less than $10\%$). In this set of experiments, each dataset is randomly partitioned into $20$ sites.
%(i.e., each data point is sent to one of the $20$ sites uniformly at random).

\ifdefined\submission
\else
Table \ref{tb:gauss} presents the experimental results on \gauss\ datasets with different $\sigma$. 
%Intuitively, increasing $\sigma$ will increase the difficulties in both clustering and outlier detection. In the task of clustering with outliers, 
We observe that \bg\ consistently gives  better \lone\ and \ltwo\ than \kmeanspar\ and \kmeanspp, and \rand\ performs the worst among all.

For outlier detection, \rand\ fails completely. In both \gauss-$0.1$ and \gauss-$0.4$, \bg\ outperforms \kmeanspp\ and \kmeanspar\ in almost all metrics. \kmeanspar\ slightly outperforms \kmeanspp. We also observe that in all \gauss\ datasets, \bg\ gives very high \prerecall, i.e., the outliers are very likely to be included in the summary constructed by \bg.

\begin{table*}[t]
\centering
\begin{tabular}{|l|l|l|l|l|l|l|l|}
  \hline
  dataset & {\tt algo} & \summarysize  &  \lone  & \ltwo &  \prerecall & \precision & \recall \\
  \hline
  \multirow{3}{*}{\gauss-$0.1$}
          & \bg & 2.40e+4 & \textbf{2.08e+5} & \textbf{4.80e+4} & \textbf{0.9890} & \textbf{0.9951} & \textbf{0.9431} \\
          & \kmeanspp & 2.40e+4 & 2.10e+5 & 5.50e+4 & 0.5740 & 0.9750 & 0.5735\\
          & \kmeanspar & 2.50e+4 & 2.10e+5 & 5.40e+4 & 0.6239 & 0.9916 & 0.6235\\
          & \rand & 2.04e+4  &2.17e+5 & 6.84e+4 & 0.0249 & 0.2727 & 0.0249\\
  \hline
  \multirow{3}{*}{\gauss-$0.4$}
          & \bg & 2.40e+4 &  \textbf{4.91e+5} & \textbf{2.72e+5} & \textbf{0.8201} & 0.7915 &  \textbf{0.7657} \\
          & \kmeanspp & 2.40e+4 & 4.97e+5 & 2.82e+5 & 0.2161 & 0.6727 & 0.2091\\
          & \kmeanspar & 2.50e+4  & 4.96e+5 & 2.79e+5 & 0.2573 &  \textbf{0.7996} & 0.2458\\
          & \rand & 2.40e+4 & 4.99e+5 & 2.90e+5 & 0.0234 & 0.2170 &  0.0212 \\
  \hline                                                          
\end{tabular}
\caption{Clustering quality on \gauss-$\sigma$ dataset, $k=100$, $t=5000$
%$\sigma$ is set as $0.1$ and $0.4$, resulting two different datasets.
}
\label{tb:gauss}
\end{table*}
\fi

\ifdefined\submission
Table \ref{tb:kdd} presents the experimental results on \kddsample\ and \kddfull\ datasets. We observe that \bg\ gives  better \lone\ and \ltwo\ than \kmeanspar\ and \kmeanspp, and \rand\ performs the worst among all.

For outlier detection, \rand\ fails completely. In both \kddfull\ and \kddsample, \bg\ outperforms \kmeanspp\ and \kmeanspar\ in almost all metrics. \kmeanspar\ slightly outperforms \kmeanspp. 

\else
Table \ref{tb:kdd} presents the experimental results on \kddsample\ and \kddfull\ datasets. In this set of experiments, \bg\ again outperforms its competitors in all metrics. Note that \kmeanspar\ does not scale to \kddfull.
\fi

\begin{table*}[t]
\centering
\begin{tabular}{|l|l|l|l|l|l|l|l|}
  \hline
  dataset & {\tt algo} & \summarysize  &  \lone  & \ltwo &  \prerecall & \precision & \recall \\
  \hline
  \multirow{3}{*}{\kddsample}
          & \bg & 3.37e+4 &  \textbf{8.00e+5} & \textbf{3.46e+6} & \textbf{0.6102} & \textbf{0.5586} & \textbf{0.5176} \\
          &\kmeanspp & 3.37e+4 & 8.38e+5 & 4.95e+6 & 0.3660 & 0.3676 &  0.1787\\        
          & \kmeanspar& 3.30e+4 & 8.18e+5 & 4.19e+6 & 0.2921 & 0.3641 & 0.1552\\
          & \rand & 3.37e+4 & 8.85e+5 & 1.06e+7 & 0.0698 & 0.5076 & 0.0374\\ 
  \hline
  \multirow{3}{*}{\kddfull}
          & \bg & 1.83e+5 & \textbf{7.38e+6} & \textbf{3.54e+7} & \textbf{0.7754} & \textbf{0.5992} & \textbf{0.5803}\\
          & \kmeanspp & 1.83e+5 & 8.21e+6 & 4.65e+7 & 0.2188 &  0.2828 & 0.1439\\
          & \kmeanspar & \multicolumn{6}{c}{does not stop after $8$ hours}\\
          & \rand & 1.83e+5 & 9.60e+6 & 1.11e+8 & 0.0378691 & 0.6115 &  0.0241\\
  \hline                                                          
\end{tabular}
\caption{Clustering quality. $k = 3$, $t = 8752$ for \kddsample\ and $t = 45747$ for \kddfull}
\label{tb:kdd}
\end{table*}

\ifdefined\submission
\else
Table \ref{tb:susy} presents the experimental results for \susy-$\Delta$ dataset. We can observe that \bg\ produces slightly better results than \kmeanspar, \kmeanspp\ and \rand\ in \lone\ and \ltwo. For outlier detection, \bg\ outperforms \kmeanspp\ and \kmeanspar\ significantly in terms of \prerecall\ and \recall, while \kmeanspar\ gives slightly better \precision. Table \ref{tb:spatial} presents the results for \spatial-$15$ dataset, and \bg\ again outperforms all other baseline algorithms in all metrics.

\begin{table*}[t]
\centering
\begin{tabular}{|l|l|l|l|l|l|l|l|}
  \hline
  dataset & {\tt algo} & \summarysize  &  \lone  & \ltwo &  \prerecall & \precision & \recall \\
  \hline
  \multirow{3}{*}{\susy-$5$}
          &\bg   & 2.40e+4 & \textbf{1.10e+7} & \textbf{2.76e+7} & \textbf{0.7508} & 0.6059 & \textbf{0.5933}\\
          &\kmeanspp & 2.40e+4 & 1.11e+7& 2.79e+7 & 0.1053 & 0.5678 & 0.1047\\
          &\kmeanspar & 2.50e+4 & 1.11e+7 & 2.77e+7 & 0.1735 & \textbf{0.7877} & 0.1609 \\
          &\rand & 2.40e+4 & 1.12e+7 & 2.84e+7 & 0.004 &  0.2080 & 0.004\\
  \hline
  \multirow{3}{*}{\susy-$10$}
          & \bg & 2.40e+4 & 1.11e+7 &  \textbf{2.77e+7} & \textbf{0.9987} & 0.9558 &  \textbf{0.9542}\\
          & \kmeanspp & 2.40e+4 & 1.11e+7 & 2.90e+7 & 0.3412 & 0.8602 & 0.3412\\
          & \kmeanspar & 2.49e+4 & 1.11e+7 & 2.84e+7 & 0.4832 & \textbf{0.9801} & 0.4823\\
          & \rand & 2.40e+4 &  1.12e+7 & 3.08e+7 & 0.0047 & 0.2481 & 0.0047\\
  \hline                                                          
\end{tabular}
\caption{Clustering quality on \susy\ dataset, $k=100$, $t=5000$}
\label{tb:susy}
\end{table*}

\begin{table*}[h]
\centering
\begin{tabular}{|l|l|l|l|l|l|l|l|}
  \hline
  dataset & {\tt algo} & \summarysize  &  \lone  & \ltwo &  \prerecall & \precision & \recall \\
  \hline
  \multirow{3}{*}{\spatial-$15$}
          &\bg   & 1.80e+3 & \textbf{5.21e+5} & \textbf{7.19e+5} & \textbf{0.9993} & \textbf{0.9993} & \textbf{0.9993}\\
          &\kmeanspp & 1.80e+3 & 5.30e+5& 7.79e+5 & 0.7698 & 0.9954 & 0.7697\\
          &\kmeanspar & 1.80e+3 & 5.28e+5 & 7.38e+5 & 0.9198 & 0.9986 & 0.9198 \\
          &\rand & 1.80e+3 & 5.35e+5 & 1.03e+6 & 0.0047 &  0.2105 & 0.0047\\
  \hline
\end{tabular}
\caption{Clustering quality on \spatial\ dataset, $k=5$, $t=400$}
\label{tb:spatial}
\end{table*}
\fi

\subsubsection{Communication Costs}
\begin{figure}[t]
  \centering
  \subfloat[][communication cost]{\includegraphics[width=4.8cm,height=3.5cm]{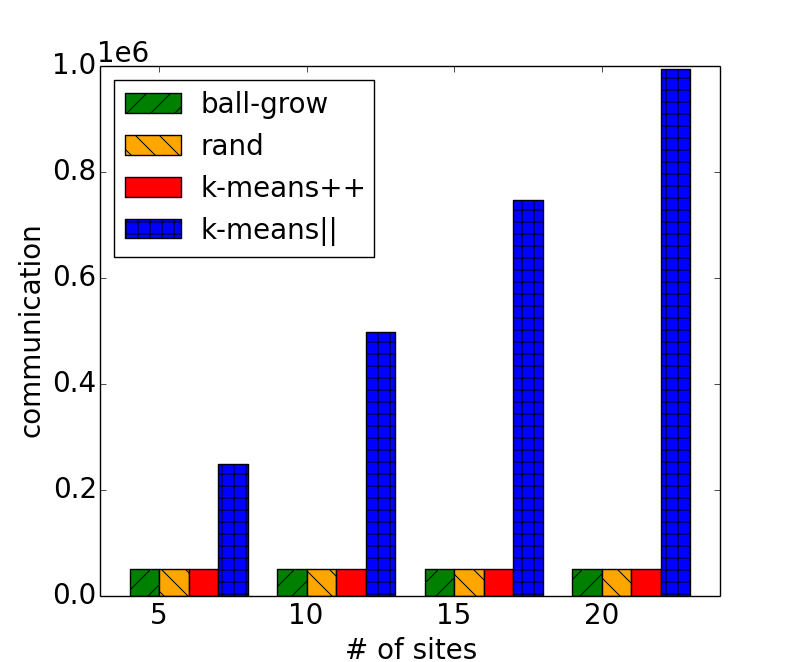}\label{fig:com}}
  \subfloat[][running time ($\log_{10}$ scale)]{\includegraphics[width=4.8cm,height=3.5cm]{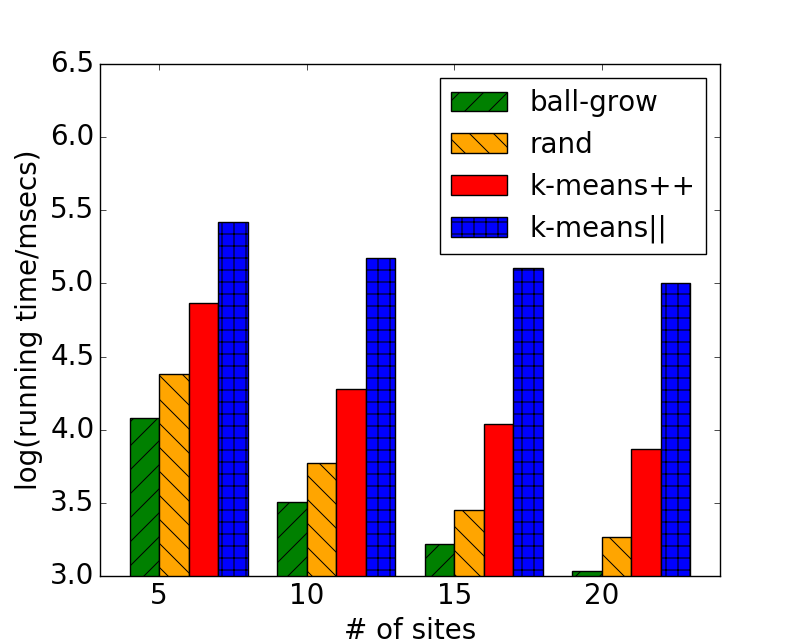}\label{fig:time}}
  \subfloat[][running time, $\#\text{sites}=20$]{\includegraphics[width=4.8cm,height=3.5cm]{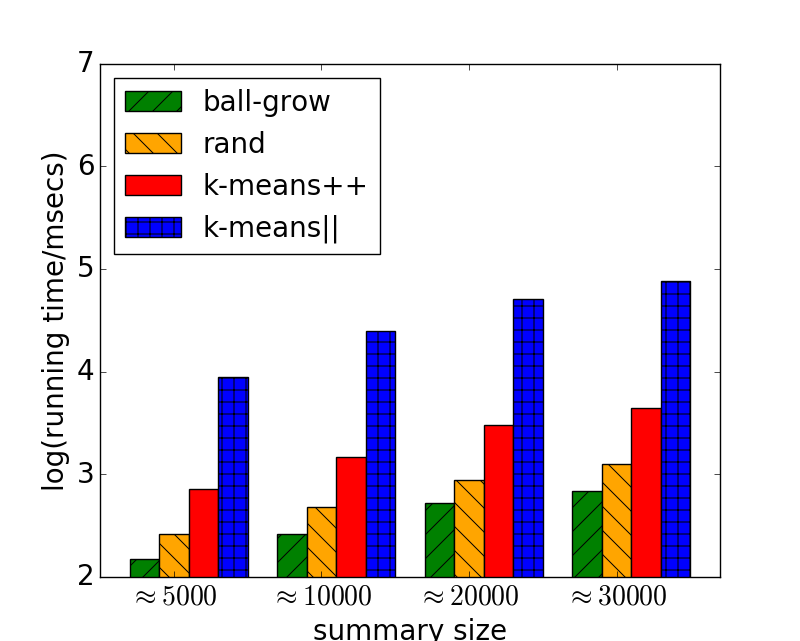}\label{fig:time-size}}
  \caption{experiments on \kddsample\ dataset}
  \label{fig:xx}
\vspace{-0.5cm}
\end{figure}
We next compare the communication cost of different algorithms. Figure \ref{fig:com} presents the experimental results. The communication cost is measured by the number of points exchanged between the coordinator and all sites. In this set of experiments we only change the number of partitions (i.e., \# of sites $s$). The summaries returned by all algorithms have almost the same size.

We observe that the communication costs of \bg, \kmeanspp\ and \rand\ are almost independent of the number of sites. Indeed, \bg, \kmeanspp\ and \rand\ all run in one round and their communication cost is simply the size of the union of the $s$ summaries. \kmeanspar\ incurs significantly more communication, and it grows almost linearly to the number of sites. This is because \kmeanspar\ grows its summary in multiple rounds; in each round, the coordinator needs to collect messages from all sites and broadcasts the union of those messages. When there are $20$ sites, \kmeanspar\ incurs $20$ times more communication cost than its competitors.

\subsubsection{Running Time}
%%\chensays{TODO: add a figure with change summary size}

We finally compare the running time of different algorithms. All experiments in this part are conducted on \kddsample\ dataset since \kmeanspar\ does not scale to \kddfull; similar results can also be observed on other datasets. The running time we show is only the time used to construct the input (i.e., the union of the $s$ summaries) for the second level clustering, and we do not include the running time of the second level clustering since it is always the same for all tested algorithms (i.e., the \kmeansmm).

Figure \ref{fig:time} shows the running time when we change the number of sites while fix the size of the summary produced by each site.  We observe that \kmeanspar\ uses significantly more time than \bg, \kmeanspp\ and \rand. This is predictable because \kmeanspar\ runs in multiple rounds and communicates more than its competitors.  \bg\ uses significantly less time than others, typically $1/25$ of \kmeanspar, $1/7$ of \kmeanspp\ and $1/2$ of \rand. The reason that \bg\ is even faster than \rand\ is that \bg\ only needs to compute weights for about half of the points in the constructed summary.
 As can be predicted, when we increase the number of sites, the total running time of each algorithm decreases.

% \begin{figure}[t]
%   \centering
%   \includegraphics[width=6cm,height=4.5cm]{figures/time-full}
% %  \includegraphics[width=7cm,height=4cm]{figures/time-part}
%   \caption{Running time ($\log_{10}$ scale) on \kddsample.} %The lower figure is identical to the upper one except that we use different scales, and \kmeanspar\ is not in the figure}
%   \label{fig:time}
% \end{figure}
% \begin{figure}[t]
%   \centering
%   \includegraphics[width=6cm,height=4.5cm]{figures/time-summary-full}
% %  \includegraphics[width=7cm,height=4cm]{figures/time-summary-part}
%   \caption{Running time ($\log_{10}$ scale) on \kddsample, \# sites is fixed to be $20$.} %The lower figure is identical to the upper one except that we use different scales, and \kmeanspar\ is not in the figure}
%   \label{fig:time-size}
% \end{figure}

%\qinsays{change \#partition to \#sites in all figures}

We also investigate how the size of the summary will affect the running time. Note that for \bg\ the summary size is controlled by the parameter $t$. We fix $k=3$ and vary $t$, resulting different summary sizes for \bg. For other algorithms, we tune the parameters so that they output summaries of similar sizes as \bg\ outputs. Figure \ref{fig:time-size} shows that when the size of summary increases, the running time increases almost linearly for all algorithms. 
%This is true for all the algorithms included in our experiment.

\ifdefined\submission
\else

\subsubsection{Stability of The Experimental Results}
Our experiments involve some randomness and we have already averaged the experimental results for multiple runs to reduce the variance. To show that the experimental results are reasonably stable, we add Table \ref{tb:kdd-stable} to present the standard deviations of the results. For each metric of a given algorithm, we gather $5$ data points,  each of which is the averaged result of 10 runs. We then calculate the mean/stddev of the $5$ data points.

\begin{table*}[h]
\centering
\begin{tabular}{|l|l|l|l|l|l|}
  \hline
   {\tt algo} &  \lone  & \ltwo &  \prerecall & \precision & \recall \\
  \hline
  %\multirow{3}{*}{\kddsample}
          \bg &  $8.16\E5 \pm 1.1\E4$ & $3.46\E6\pm 4.1\E5$ & $0.61\pm 0.002$ & $0.55 \pm 0.007$ & $0.52\pm 0.004$ \\
          \kmeanspp & $8.83\E5 \pm 6.9\E 4$ & $5.11\E6 \pm 2.8\E5$ & $0.37\pm 0.004 $ & $0.36 \pm 0.004$ &  $0.18 \pm 0.002$\\        
           \kmeanspar & $8.41\E5 \pm 5.0 \E4$ & $4.19\E6 \pm 1.4\E5$ & $0.29\pm 0.004$ & $0.36\pm 0.005$ & $0.16 \pm 0.004$\\
           \rand  & $9.20\E5 \pm 5.9\E4$ & $1.08\E7 \pm 2.0\E5$ & $0.07\pm 0.001$ & $0.49\pm 0.009$ & $0.04\pm 0.005$\\ 
  \hline                                                          
\end{tabular}
\caption{Clustering quality on \kddsample\, $k = 3$, $t = 8752$. Each entry is in the format of mean$\pm$stddev.}
\label{tb:kdd-stable}
\end{table*}

It can be seen from Table \ref{tb:kdd-stable} that the results of our experiments are very stable in almost all metrics. \lone\ is the only metric where our algorithm has some overlap with other baseline algorithms, but it is still safe to  conclude that our algorithm outperforms all the baselines in almost all metrics. The similar stability is observed on other datasets.

\subsubsection{Summary}
We observe that \bg\ gives the best performance in almost all metrics for measuring summary quality. \kmeanspar\ slightly outperforms \kmeanspp. \rand\ fails completely in the task of outliers detection.   For communication, \bg, \kmeanspp\ and \rand\ incur similar costs and are independent of the number of sites. \kmeanspar\ communicates significantly more than others. For running time, \bg\ runs much faster than others, while \kmeanspar\ cannot scale to large-scale datasets. 
%in our implementation.
%which is mainly because it runs in multiple rounds and needs to communicate more messages.
\fi
%Our experimental results are strong evidences that our proposed algorithm can be very practical and useful in real-world applications.

%\subsubsection{Some Implementation Details}
%\chensays{not sure if we want to mention this in the submission. Maybe we can add it later after the paper gets accepted?}
%Here we briefly describe how \kmeansmm\ \cite{CG13} works and how we can extend it to the weighted setting. \kmeansmm\ works as follows: it first randomly samples $k$ points uniformly at random from the dataset as the initial centers, and then it goes round by round. In each round, \kmeansmm\ finds the $t$ points that are farthest from the current centers as the outliers, and assigns each of the remaining point to its closest center. Each cluster (i.e. the group of points that are assigned to the same center) then averages all of its points to get the new center. The algorithm stops after the set of centers does not change (too much) between two successive rounds. 
%
%In the weighted setting, when each cluster averages all its points, it takes the weights into account. Also instead of finding $t$ different points as the outliers, we pick the farthest points with a total weight $t$. Note that one point might be split into two in this process.

\subsubsection*{Acknowledgments}
Jiecao Chen, Erfan Sadeqi Azer and Qin Zhang are supported in part by
NSF CCF-1525024 and IIS-1633215.

\ifdefined\submission
%\clearpage
\else
\fi

\bibliography{paper}
\bibliographystyle{natbib}
\end{document}